    \documentclass{article}
    \usepackage{amsthm}

\title{\LARGE \bf
Asynchronous opinion dynamics on the $k$-nearest-neighbors graph
}

\author{Wilbert Samuel Rossi and Paolo Frasca
\thanks{W.S.\ Rossi and P.\ Frasca are with Department of Applied Mathematics,
        University of Twente, 7500 AE Enschede, The Netherlands 
        {\tt\small w.s.rossi@utwente.nl}. P.\ Frasca is with Univ.\ Grenoble Alpes, CNRS, Inria, Grenoble INP, GIPSA-lab, 
	F-38000 Grenoble, France
        {\tt\small paolo.frasca@gipsa-lab.fr}}%
\thanks{This work has been partly supported by IDEX Universit\'e Grenoble Alpes under C2S2 ``Strategic Research Initiative'' grant. The authors also acknowledge the inspiring conversations with J.M.\ Hendrickx and S.\ Martin.}
}

\usepackage{graphicx,xcolor}
\graphicspath{{./figure/}}

\newlength{\figtriml}
\newlength{\figtrimb}
\newlength{\figtrimr}
\newlength{\figtrimt}
\newlength{\figwidth}
\newlength{\figtrimbX}
\newlength{\figtrimtX}
\usepackage{amsmath,amssymb}

    \newtheorem{theorem}{Theorem}

    \newtheorem{lemma}[theorem]{Lemma}
    \newtheorem{coro}[theorem]{Corollary}
   
   \newtheorem{example}{Example}

\newcommand{\R}{\mathbb{R}}

\renewcommand{\P}{\mathbb{P}}

\newcommand{\xx}{\mathbf{x}}

\newcommand{\1}{\mathbf{1}}

\newcommand{\mx}{{\mu(\xx)}}
\newcommand{\mxz}{{\mu(\xx(0))}}

\newcommand{\mxt}{{\mu(\xx(t))}}
\newcommand{\mxtpo}{{\mu(\xx(t+1))}}

\newcommand{\Mx}{{M(\xx)}}

\newcommand{\Mxt}{{M(\xx(t))}}

\newcommand{\kmo}{{(R)}}
\newcommand{\dkmo}{{(T)}}

\begin{document}

\maketitle
\thispagestyle{empty}
\pagestyle{empty}

\begin{abstract}
This paper is about a new model of opinion dynamics with opinion-dependent connectivity. We assume that agents update their opinions asynchronously and that each agent's new opinion depends on the opinions of the $k$ agents that are closest to it. We show that the resulting dynamics is substantially different from comparable models in the literature, such as bounded-confidence models. We study the equilibria of the dynamics, observing that they are robust to perturbations caused by the introduction of new agents. We also prove that if the number of agents $n$ is smaller than $2k$, the dynamics converge to consensus. This condition is only sufficient.
\end{abstract}

\section{Introduction}

Driven by the evolution of digital communication, there is an increasing interest for mathematical models of opinion dynamics in social networks. A few such models have become popular in the control community, see the surveys  
\cite{PT:2017,PT:2018:part2}. In the perspective of the control community, opinion dynamics distinguish themselves from consensus dynamics because consensus is prevented by some other dynamical feature. In many popular models, this feature is an opinion dependent limitation of the connectivity. This is the case of bounded confidence (BC) models~\cite{krause:2000:discrete,Deffuant:2000}, where social agents influence each other iff their opinions are closer than a threshold. 
This way of defining connectivity, however, seems at odds with several social situations, since it may require an agent to be influenced by an unbounded number of fellow agents. Instead, the number of possible interactions is capped in practice by the limited capability of attention by the individuals. For instance, online social network services are based on recommender systems that select a certain number of news items, those which are closer to the user's presumed tastes. 
However, to the best of our knowledge, this important observation has not been incorporated in any suitable model of opinion dynamics, with the partial exception of~\cite{Piccoli:2017:chapter}. The latter paper compares different models of interaction, including one in which each agent is influenced by a fixed number of neighbors.

In a striking contrast, this observation has been made in the field of biology by a number of quantitative studies about flocking in animal groups (these include both theoretical and experimental works) 
\cite{Ballerini:2008:evidence,Giardina:2008:collective,Frasca:2011:animal-anisotropic,Frasca:2017:birds-wires}.
The importance of this way of defining connectivity has been also captured by graph theorists, who have studied a the properties of what they call $k$-nearest-neighbors graph. For instance, it is known that $k$ must be logarithmic in $n$ to ensure connectivity~\cite{Balister:2005:connectivity} and flocking behavior \cite{Guo:2017:consensus-flocks}.

In this paper, we provide the first analysis of the $k$-nearest-neighbor opinion dynamics. In this analysis, our contribution is threefold: (1) We describe the equilibria of the dynamics, distinguishing a special type of {\em clustered equilibria} that are constituted of separate clusters; (2) We discuss the robustness of clustered equilibria to perturbations consisting in the addition of new agents; (3) We provide a proof of convergence for small groups, that is, groups such that $n<2k$.
%

Our work differs from~\cite{Piccoli:2017:chapter} in several aspects. As per the model, the dynamical model in~\cite{Piccoli:2017:chapter} is synchronous and continuous-time, whereas ours is asynchronous and discrete-time. As per the analysis, \cite{Piccoli:2017:chapter} focuses on the equilibria and their properties (for instance, the distribution of their clusters' sizes) are studied by extensive simulations, whereas we study the dynamical properties (robustness to perturbations, convergence) by a mix of simulations and analytical results. 
Our robustness analysis is based on the approach taken by Blondel, Hendrickx and Tsitsiklis for BC models~\cite{VDB-JMH-JNT:09}.
Our convergence result is inspired by classical proofs of convergence for randomized consensus dynamics~\cite[Chapter~3]{FF-PF:17}, but its interest and difficulty originate from the lack of reciprocity in the interactions:  this feature clearly distinguishes our model from bounded confidence models, where interactions are reciprocal as long as the interaction thresholds are equal for all agents~\cite{krause:2000:discrete,Piccoli:2017:chapter,VDB-JMH-JNT:09a,AM-FB:11f,CC-FF-PT:12,FC-PF:11}.



\section{The dynamical model}

%
Let $n$ and $k$ be two integers with
$$1 \leq k \leq n,$$ 
and  let $V = \{1,\ldots,n\}$ be the set of agents.  
Each agent is endowed with a scalar opinion $x_i \in \R$, to be updated asynchronously. 
The update law
\begin{equation} \label{eq:model}
\xx^+ = f(\xx,i) 
\end{equation}
goes as follows. An agent $i$ is selected from $V$; the elements of $V$ are ordered by increasing values of $|x_j-x_i|$; then, the first $k$ elements of the list (i.e. those with  smallest distance from $i$) form the set $N_i$ of current neighbors of $i$. 
Should a tie between two or more agents arise, priority is given to agents with lower index.
Agent $i$ may but not necessarily does belong to $N_i$. 
Once $N_i$ is determined, agent $i$ updates his opinion $x_i$  to
$$x_i^+ =  \frac{1}{k} \sum_{j\in N_i} x_j\,,$$
while the remaining agents do not change their opinions
$$x_{j}^+ = x_{j}\quad \text{for every }j\neq i\,.$$


We show a couple of simulations to illustrate the possible behaviours of the model, see Figure~\ref{fig:sim-typ} and \ref{fig:sim-non-clust}. For these simulations we set $n=20$, $k=5$ and choose the initial opinion of every agent uniformly at random in $[0,1]$. At every step, we choose from $V$ the node that updates opinion, independently and uniformly at random. The simulation of Figure~\ref{fig:sim-typ} shows a typical outcome: the agents form two distinct groups (of 10 agents each) with homogeneous opinions; for every agent, his neighbors at time $t=1000$ have almost the same opinion. 
This last observation does not hold in the simulation of Figure~\ref{fig:sim-non-clust}: the two pairs of agents that at time $t=1000$ have opinion about 0.6 and about 0.7, respectively, have neighbors with different opinions. 
These distinct behaviors lead us to distinguish different kinds of equilibria: this will be the topic of the next section.

\begin{figure} 	
\centering 	
\fbox{%
\includegraphics[trim={\figtriml} {\figtrimb} {\figtrimr} {\figtrimt}, 	clip, width={\figwidth}, keepaspectratio=true]{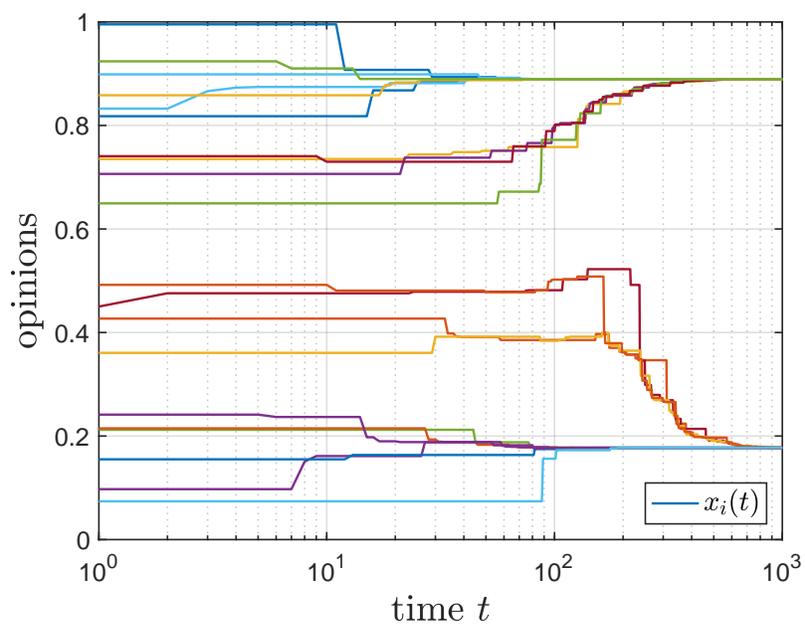}%
	}
 	\caption{\label{fig:sim-typ} Simulation of the model \eqref{eq:model} with $n=20$, $k=50$, initial opinions chosen uniformly at random in $[0,1]$ and update sequence chosen uniformly at random. The plot contains a typical trajectory that converges to a clustered equilibrium.  }
\end{figure}

\begin{figure} 	
\centering 	
\fbox{%
\includegraphics[trim={\figtriml} {\figtrimb} {\figtrimr} {\figtrimt}, 	clip, width={\figwidth}, keepaspectratio=true]{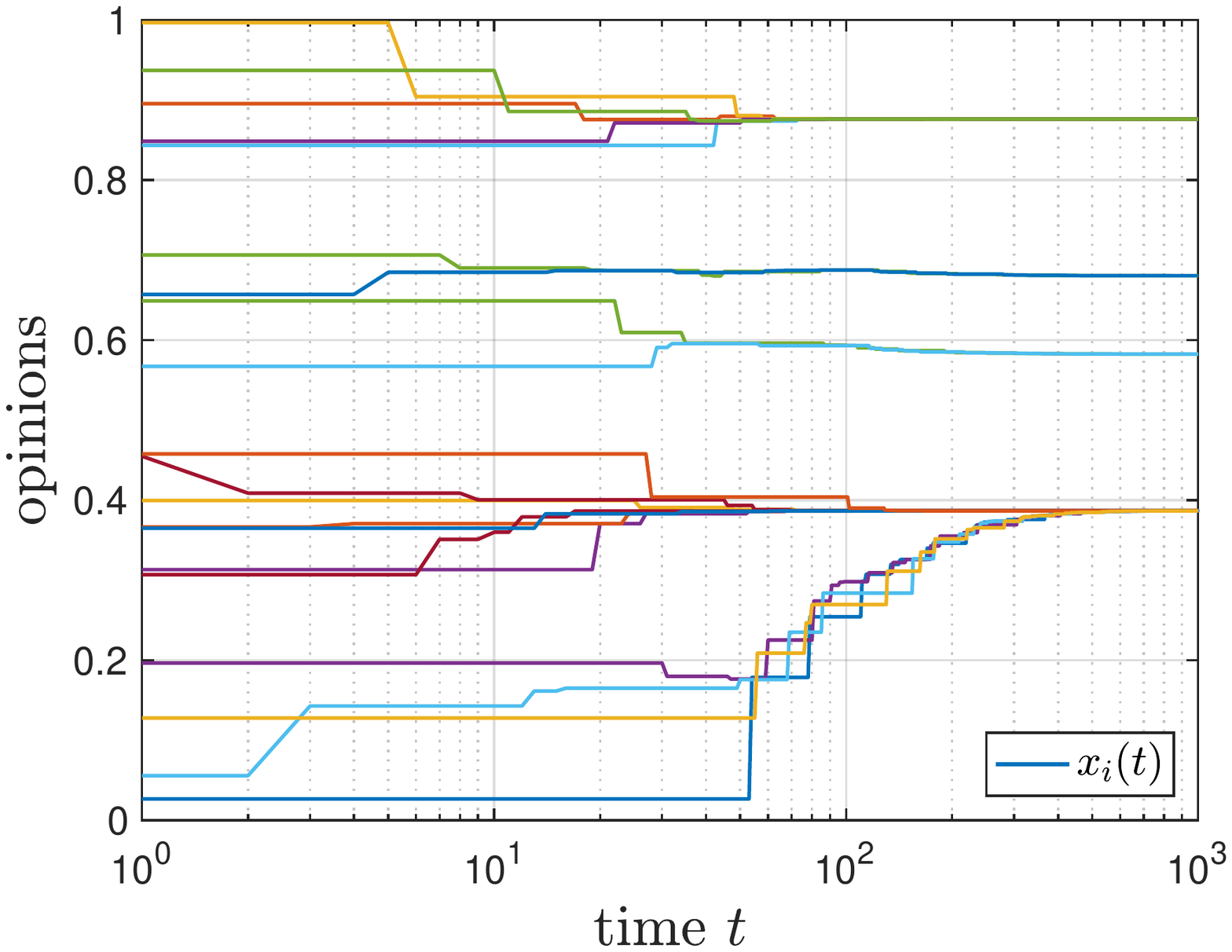}%
	}
 	\caption{\label{fig:sim-non-clust} Simulation of the model \eqref{eq:model} with $n=20$, $k=50$, initial opinions chosen uniformly at random in $[0,1]$ and update sequence chosen uniformly at random. The plot contains a less common trajectory that converges to a non-clustered equilibrium. }
\end{figure}

\section{Equilibria}\label{sec:equilibria}
In this section we discuss some properties of the equilibria of system~\eqref{eq:model}. 
Motivated by the simulations, we introduce the following terminology.
%
Given a configuration $\xx \in \R^n$, the directed graph that represents the possible interactions (i.e. the opinion dependancies for any possible selection of the node to be updated) is 
$$G(\xx) = (V, E(\xx)) \quad \text{with}\quad E(\xx) = \bigcup_{i\in V}  \{ (i,j), j\in N_i\}\,, $$
where $N_i$ is the set of neighbors of $i$, should $i$ be selected to update his opinion. Clearly, if $k=n$ the graph $G(\xx) = (V, V\times V)$ is complete. 
A configuration $\xx \in \R^n$ is an \emph{equilibrium} for the asynchronous dynamics if 
$$\xx = f(\xx,i) \quad \text{for every } i\,. $$
If $k=1$, then $G(\xx)$ contains only links between nodes with the same opinion: 
in this trivial case, every configuration is an equilibrium because agents cannot change opinion. 

A configuration $\xx$ is called \emph{clustered} if 
$$\xx_{N_i} = x_i \1_{N_i} \quad \text{for every } i\,,$$
that is, if for every node all of his neighbors have the same opinion.
Furthermore, a clustered configuration $\xx = c\1$ for some $c\in \R$ is called \textit{consensus}.

%

%
%



It is immediate to see that clustered configurations are equilibria.
However, there exist equilibria that are not clustered. It is possible to obtain a simple counterexample with $n=7$ and $k=3$ and exploiting the tie break rule. Consider any configuration $\xx \in \R^7$ of the form 
$$\xx_{\{1,3,5\}} = \alpha\,\1_{\{1,3,5\}}\,,\quad\xx_{\{2,4,6\}} = \beta\,\1_{\{2,4,6\}}\,,\quad x_7 = \tfrac{\alpha+\beta}2\,,$$ 
where $\alpha,\beta\in \R$ and $\alpha<\beta$.
The above is an equilibrium even if $\xx_{N_7}=\xx_{\{1,2,7\}} \neq \frac12(\alpha+\beta)\1_{\{1,2,7\}}$. 

The tie breaking rule is not central for the existence of non-clustered equilibria, as one can see in the following example inspired by Figure~\ref{fig:sim-non-clust}. 
\begin{example}
Consider $\xx \in \R^{20}$ with 
\begin{align*}
&\xx_{\{1,2,\ldots,11\}} = \alpha\,\1_{\{1,2,\ldots,11\}}\,, \\
&x_{12} = x_{13} = \tfrac{3\alpha+2\beta}5\,,\\ 
&x_{14} = x_{15} = \tfrac{2\alpha+3\beta}5\,,\\
&\xx_{\{16,17,\ldots,20\}} = \beta\,\1_{\{16,17,\ldots,20\}}\,,
\end{align*}
where $\alpha,\beta\in \R$ and $\alpha < \beta$.
For instance, the neighbors of agent 12 are $N_{12} = \{1,12,13,14,15\}$ because
\begin{align*}
&|x_{12}-x_{12}| = |x_{13}-x_{12}| =0\,, \\
&|x_{12}-x_{14}| = |x_{12}-x_{15}| =\tfrac{1}5(\beta-\alpha)\,, \\
&|x_{12}-x_1| = \tfrac{2}5(\beta-\alpha)\,, 
\end{align*}
while the remaining agents are at distance $\tfrac{2}5(\beta-\alpha)$ or larger. 
Such configuration is an equilibrium with $\xx_{N_{12}} \neq x_{12}\,\1_{N_{12}}$.
\end{example}\smallskip

A simple analysis shows that clustered configuration are those in which the agents form \textit{clusters} of at least $k$ participants with the same opinion.
To make this claim formal, let $V_i = \{j:x_j = x_i\}$ be the set of nodes that share the same opinion of $i$.
\begin{lemma}\label{lem:cluster-size}
A configuration is clustered if and only if $|V_i| \ge k $ for every $i$.
\end{lemma}
\begin{proof}
By definition, in a clustered configuration $N_i \subseteq V_i$ for every $i$.
Assume $|V_i| \ge k $ for every $i$. For any $i$ there are at least $k$ nodes $j$ (including $i$) with $x_j = x_i$: such nodes have zero distance from $i$ and hence $N_i \subseteq V_i$. 
This holds for every $i$ so the configuration is clustered.
On the other hand, assume that exists $i$ with $|V_i| \leq k-1 $. The set $N_i$ must contain a node $j$ with $x_j\neq x_i$ so not in $V_i$, violating the definition of clustered configuration. 
\end{proof}


From this result, it follows that a clustered configuration allows up to $$\left\lfloor \frac{n}{k} \right\rfloor $$ 
distinct sets $V_i$ (and this bound is tight). For the special case of consensus, this claim becomes the following corollary.
\begin{coro}
Consensus is the only possible clustered configuration if and only if $$n < 2k\,.$$
\end{coro}

\section{ROBUSTNESS OF THE EQUILIBRIA}


The clustered equilibria of the dynamics described above have interesting robustness properties regarding the addition of new nodes or the removal of nodes. The model shows different behavior with respect to a standard 
Asynchronous Bounded Confidence (ABC) model.
In this section, we briefly introduce for comparison the ABC model; then we provide a few simulations to motivate the following discussion of the robustness properties.

\subsection{ABC model}
Given a fixed \textit{range of confidence} $d>0$, we introduce the Asynchronous Bounded Confidence (ABC) update law 
\begin{equation}\label{eq:ABC-model}
\xx^+ = f_{ABC}(\xx,i)\,. \end{equation}
where $i$ is the agent that updates his opinion.
The neighborhood of $i$ is
$N^{ABC}_i = \{j: |x_j-x_i|\le d \}$
and always contains $i$ itself. The new opinion of agent $i$ is 
$$x_i^+ =  \frac{1}{|N^{ABC}_i|} \sum_{j\in N^{ABC}_i} x_j\,,$$
while the remaining agents do do not change opinion
$$x_{j}^+ = x_{j}\quad \text{for every }j\neq i\,.$$

\subsection{Simulations}

\begin{figure} 	
\centering 	
\fbox{%
\includegraphics[trim={\figtriml} {\figtrimbX} {\figtrimr} {\figtrimtX}, clip, width={\figwidth}, keepaspectratio=true]{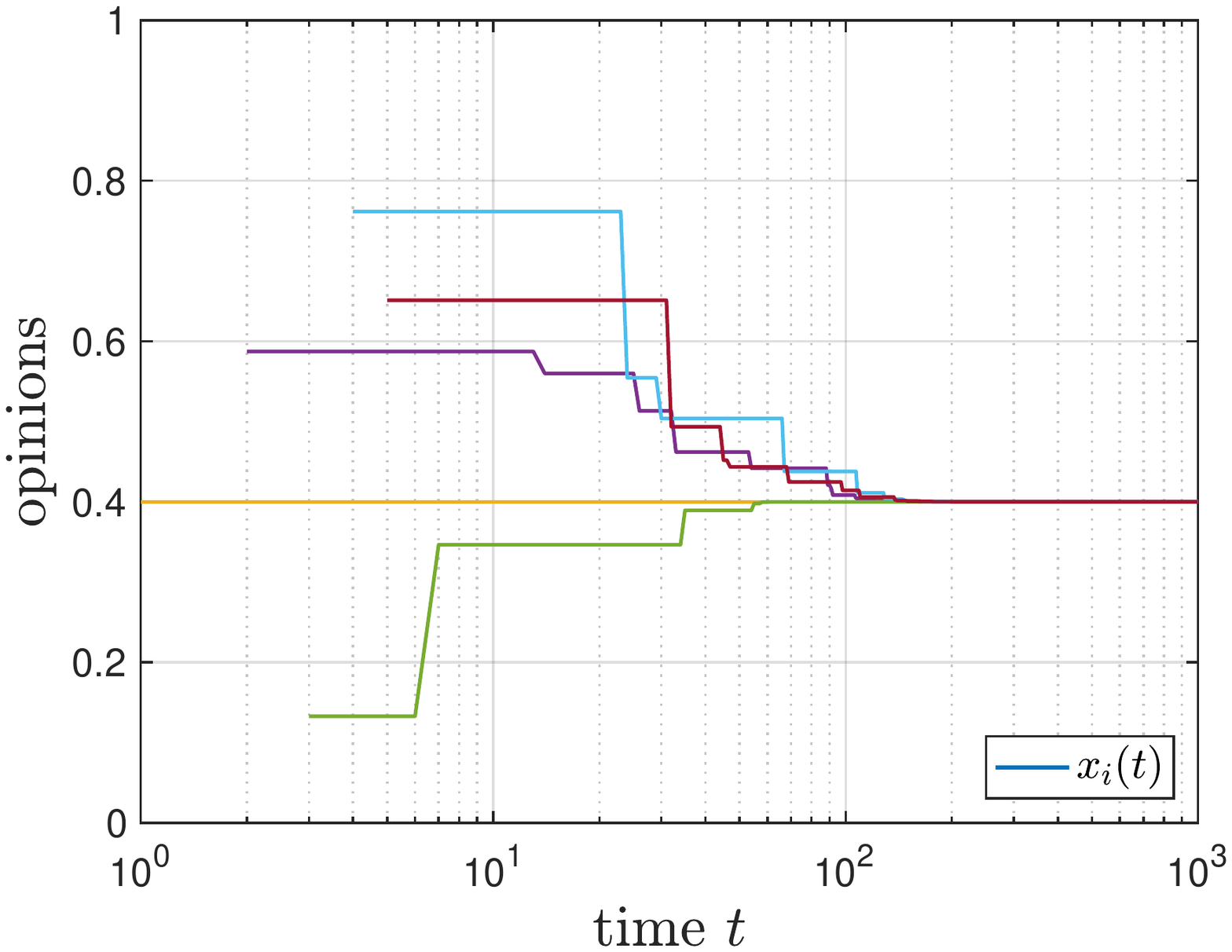}%
	}

\fbox{%
\includegraphics[trim={\figtriml} {\figtrimbX} {\figtrimr} {\figtrimtX}, clip, width={\figwidth}, keepaspectratio=true]{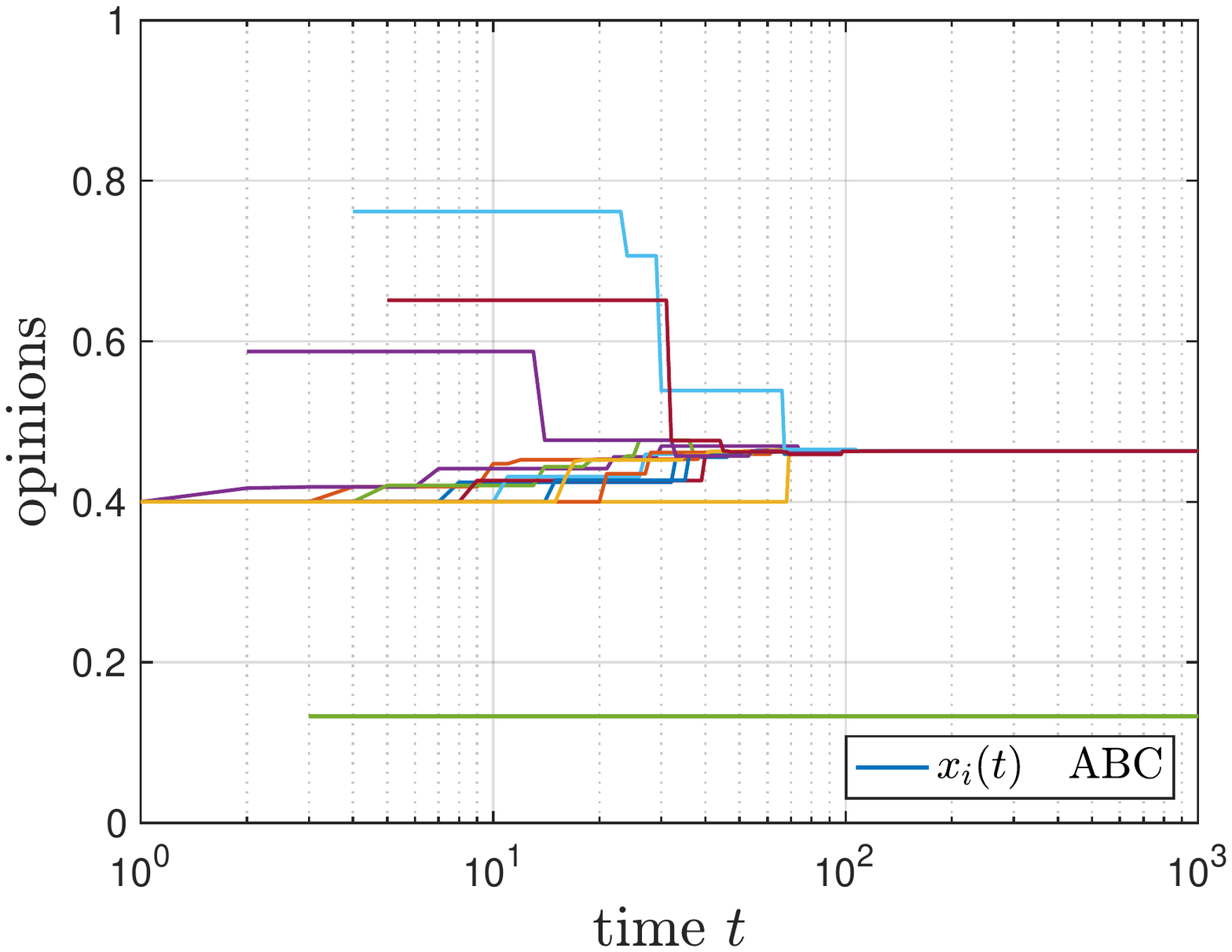}%
	}
 	\caption{\label{fig:robu-add-4} The addition of four new nodes to a consensus configuration with ten nodes. Upper plot: the trajectory of the model \eqref{eq:model} with $k=5$. Lower plot: the trajectory of the model \eqref{eq:ABC-model} with $d=0.25$. The same initial conditions and update order are used.  }
\end{figure}

We present a simulation to show the difference between model~\eqref{eq:model} and model~\eqref{eq:ABC-model} when a few agents are added to a consensus configuration (which is an equilibrium for both models).
We set $k=5$ for model \eqref{eq:model} and $d=0.25$ for model \eqref{eq:ABC-model}. We start with 10 agents sharing opinion $0.4$; at steps $t=2,3,4,5$ we add a new agent, with opinion chosen uniformly at random in $[0,1]$.
We select the agent that updates his opinion among those present at that time, independently and uniformly at random: the same selection is used in both models. 
Figure~\ref{fig:robu-add-4} contains the plots of the simulation. 
The upper plot regards the dynamics of model \eqref{eq:model}: the four new agents converge to the consensus opinion, which does not change; they are too few to form a new cluster. 
The lower plot contains the dynamics of the model \eqref{eq:ABC-model}: the consensus configuration is not preserved and the agent added at step $t=3$ remains isolated during the dynamics and keeps his opinion. The other three new agents join the original ten; this group of 13 agents converge to the same opinion which however is different from the original consensus value.

\subsection{Robustness of the equilibria}
We now provide a general discussion that explains the observations from Figure~\ref{fig:robu-add-4}.
Let $n,k$ with $1\le k\le n$ be given and consider a clustered equilibria $\xx\in \R^n$ of the model \eqref{eq:model}. 
We first discuss the {\em addition} of a new agent with opinion $x_{n+1} = \alpha$ to the configuration $\xx$, that becomes  $[\xx;\alpha]\in \R^{n'}$ with $n'=n+1$.
Before the addition of the new node, clusters have to contain at least $k$ agents. 
This fact remains true after the addition and we have that
$$f([\xx;\alpha],i)  = [\xx;\alpha']$$
for every $i$, meaning that the original (clustered) portion of the configuration $[\xx;\alpha]$ remains unperturbed. 
For a generic value of $\alpha$ 
the limit of the dynamics has the same cluster locations of $\xx$, with one of the clusters getting a new member. For some specific values, it may happen that the configuration $[\xx;\alpha]$ is a non-clustered equilibrium. 
%
In any case, none of the original agents changes opinion.  
Instead, in the metric ABC model \eqref{eq:ABC-model} with uniform visibility radius $d$, either the new agent is further apart from the original agents and nothing happens or he falls within the visibility radius of a cluster of agents. In the latter case both the new agents and the agents in the cluster change opinions, converging to an intermediate value.

Assuming $n$ sufficiently large, the {\em removal} of an agent from a clustered equilibrium presents interesting differences too. 
In the metric ABC model \eqref{eq:ABC-model} the removal of an agent does not trigger any dynamics in the remaining agents. 
In model \eqref{eq:model}, if the agent is removed from a cluster with $k+1$ agents or more, nothing happens. But if the agent is removed from a cluster with $k$ agents, the new configuration is not an equilibrium anymore and the remaining nodes from that group will evolve towards some new equilibrium.

\addtolength{\textheight}{-3cm}   

%
%

\section{Convergence to consensus}
In this section we show that process~\eqref{eq:model} converges to a consensus, provided $n<2k$ and the choice of the agent that updates his opinion at time $t$ is an i.i.d. uniform random variable over $V$. 
We recall from Section~\ref{sec:equilibria} that the consensus is the unique clustered equilibrium for $n<2k$.

For $t\ge 0$, let $\xx(t) \in \R^n$ be the sequence of opinion vectors and $I(t) \in V$ a sequence of agents. 
Given an initial configuration $\xx(0) = \xx^0$, we consider the dynamics
\begin{equation}\label{eq:dyn-special}
\xx(t+1) = f(\xx(t) , I(t)) \quad \text{for every } t\ge 0\,,
\end{equation}
where $I(t)$ is the agent that updates his opinion at time $t$. 
%


%
%

We introduce two functions $\mu, M:\R^n \to V$ that, given an opinion vector $\xx$, return respectively the index of the smallest and largest components, with ties sorted
$$\mu(\xx) = \min( \arg \min_i x_i )\,, \qquad M(\xx) = \min( \arg \max_i x_i )\,.$$
The outer $\min$ sorts possible ties; note that $M(\xx) = \mu(-\xx)$.

In the following two lemmas we prove the properties of the dynamics in which the agent with smallest opinion is the one that updates his opinion. 
\begin{lemma}\label{lem:prelim-mu}
Given $n,k$ with $1\leq k\leq n$ and an initial configuration $\xx^0\in \R^n$ consider dynamics~\eqref{eq:dyn-special} with $I(t) = \mu(\xx(t))$ and the scalar sequence $y(t) := \max_{i\in N_{\mu(\xx(t))}} x_i(t)$. 
Then:
\begin{itemize}
\item the set sequence $N_{\mu(\xx(t))} $ and the scalar sequence $y(t)$ are constant;
\item for every $i \in N_{\mu(\xx(0))} $ the sequences $x_i(t)$ are non-decreasing   and satisfy $x_i(t) \leq y(0)$; 
\item  for every $i \notin N_{\mu(\xx(0))} $ the sequences $x_i(t)$ are constant.
\end{itemize}
\end{lemma}
\begin{proof} The proof goes by induction. First, consider the trivial case with $x_{\mu(\xx(t))}(t) = y(t)$. This condition means $x_i(t) = y(t)$ for every $i\in N_{\mu(\xx(t))}$ and thus $x_{\mu(\xx(t))}(t+1) = x_{\mu(\xx(t))}(t)$ so  everything remains unchanged. 

Next, consider the case with $x_{\mu(\xx(t))}(t) < y(t)$. We have
$$x_{\mu(\xx(t))}(t+1) = \frac1k \sum_{j\in N_{\mu(\xx(t))}} x_j(t) \in \big(x_{\mu(\xx(t))}(t) , y(t) \big) \,.$$
Therefore,  
$$ \{i: x_i(t) < y(t) \} = \{i: x_i(t+1) < y(t) \} $$
and 
$$ \{i: x_i(t) = y(t) \} = \{i: x_i(t+1) = y(t) \}\,. $$
Moreover, the cardinality of the set $\{i: x_i(t) < y(t) \}$ is strictly smaller than $k$. This implies that $N_\mxtpo = N_\mxt$ and also $y(t+1) = y(t)$. 
The claims follow by induction and by observing that only the agents $i\in N_\mxz$ can update their opinions at some time $t\ge0$ and the updated value $x_i(t+1)$ belongs to $[x_i(t),y(t)]$.
\end{proof}

\begin{lemma}\label{lem:shrink-mu}
Given $n,k$ with $1\leq k\leq n$ and an initial configuration $\xx^0\in \R^n$ consider the dynamics \eqref{eq:dyn-special} with $I(t) = \mu(\xx(t))$ and the scalar sequence $y(t) = \max_{i\in N_{\mu(\xx(t))}} x_i(t)$. 
Then
$$y(k\!-\!1) - \min_i x_i(k\!-\!1) \leq \left( 1-\tfrac1k \right)\!\big(y(0) - \min_i x_i(0) \big)$$
\end{lemma}
\begin{proof} First, compute $x_\mxt(t+1)$ for a generic $t\ge0$. 
We have
\begin{align*}
x_\mxt(t+1) &= {\textstyle \frac1k \sum_{j\in N_\mxt} x_j(t) }\\
&={\textstyle\frac1k \sum_{j\in N_\mxz} x_j(t) }\\
&\ge {\textstyle \frac1k \sum_{j\in N_\mxz} x_j(0)}
\end{align*}
thanks to Lemma~\ref{lem:prelim-mu}. Then,
\begin{align*}
x_\mxt(t+1) &\ge \tfrac{k-1}k x_\mxz(0) + \tfrac1k y(0) \\
&= x_\mxz(0)  + \tfrac1k \big(y(0)-  x_\mxz(0)\big) \,.
\end{align*}
Next, consider the set
$$S(t) = \big\{i: x_i(t) < x_\mxz(0)  + {\textstyle \frac1k} \big(y(0)-  x_\mxz(0)\big) \big\} \,,$$
and observe that either $S(t)=\emptyset$ 
or $|S(t\!+\!1)| = |S(t)|-1$ because $\mxt \notin S(t+1)$. Since the set $S(0)$ contains at most $k-1$ elements, the set $S(k\!-\!1)$ is empty. 
Hence, 
$$x_i(k-1)\ge x_\mxz(0)  + {\textstyle \frac1k} \big(y(0)-  x_\mxz(0)\big)$$ 
for every $i$, a fact that implies  
$$x_{\mu(\xx(k-1))}(k-1)\ge x_\mxz(0)  + {\textstyle \frac1k} \big(y(0)-  x_\mxz(0)\big)\,.$$

Using Lemma~\ref{lem:prelim-mu} we know that $N_\mxt = N_\mxz$ for every $t\ge 0$ and that for every $i$ therein, $x_i(t)\leq y(t) = y(0)$.
Therefore 
\begin{align*}
y(k\!-\!1) - x_{\mu(\xx(k-1))}(k\!-\!1)\le &\,y(0) - x_\mxz(0)\,  \\
&- {\textstyle \frac1k} \big(y(0)-  x_\mxz(0)\big)
\end{align*}
and the thesis follows because $x_\mxt = \min_i x_i(t)$.
\end{proof}

The following lemma follows from Lemma~\ref{lem:prelim-mu} and \ref{lem:shrink-mu} using the property $M(\xx) = \mu(-\xx)$.
\begin{lemma}\label{lem:shrink-prelim-M} 
Given $n,k$ with $1\leq k\leq n$ and an initial configuration $\xx^0\in \R^n$ consider the dynamics \eqref{eq:dyn-special} with $I(t) = M(\xx(t))$ and the scalar sequence $z(t) := \min_{i\in N_{M(\xx(t))}} x_i(t)$. 
Then:
\begin{itemize}
\item the set sequence $N_{M(\xx(t))} $ and the scalar sequence $z(t)$ are constant;
\item for every $i \in N_{M(\xx(0))} $ the sequences $x_i(t)$ are non-increasing   and satisfy $x_i(t) \geq z(0)$; 
\item for every $i \notin N_{M(\xx(0))} $  the sequences $x_i(t)$ are constant.
\end{itemize}
Moreover,
$$ \max_i x(k\!-\!1) - z(k\!-\!1) \leq \left( 1-\tfrac1k \right) \!\big(\max_i x_i(0) - z(0) \big)\,.$$
\end{lemma}

\smallskip
The next equivalence will be crucial in the following. 
\begin{lemma}\label{lem:z-y}
Given $n,k$ with $1\leq k\leq n$, consider $\xx \in \R^n$ and define the quantities
$$y := \max_{i\in N_{\mu(\xx)}} x_i \qquad \text{and} \qquad z := \min_{i\in N_{M(\xx)}} x_i\,.$$
Then, $z \leq y$ for every $\xx \in \R^n$ if and only if $n<2k$. 
\end{lemma}
\begin{proof} 
We prove the equivalent claim that  $\xx \in \R^n$ with $z > y$ exists if and only if  $n\ge2k$.
Indeed, if $n \ge 2k$ consider the vector $\xx \in \R^n$ such that 
$$x_1\leq x_2 \leq \ldots \leq x_k < x_{k+1}\leq \ldots \leq x_{n-k+1} \leq \ldots \leq x_n $$
where $n-k+1 > k$. The set $N_\mx$ contains the $k$ smallest elements of $\xx$ so $y=x_k$, while the set $N_\Mx$ contains the $k$ largest  elements of $\xx$, so $z=x_{n-k+1}>x_k = y$. 
For the converse, assume that $\xx$ with $z>y$ exists, meaning  
$${\textstyle \left(\max_{i\in N_{\mu(\xx)}} x_i \right) < \left( \min_{i\in N_{M(\xx)}} x_i\right).}$$ 
Both sets $N_\mx$ and $N_{M(\xx)}$ contain $k$ elements, so the sets 
$${\textstyle \{j : x_j \leq \max_{i\in N_{\mu(\xx)}} x_i \}}\text{~ and ~}
{\textstyle\{j : x_j \geq \min_{i\in N_\Mx} x_i \} }$$
 contain at least $k$ elements each. These two sets are disjoint, thus the vector $\xx\in \R^n$ has at least $n\ge 2k$ components. 
\end{proof}

The next lemma describes a ``shrinking sequence''. 
\begin{lemma}\label{lem:special-sequence}
Given $n,k$ with $1\leq k\leq n$ and an initial configuration $\xx^0\in \R^n$ consider the dynamics \eqref{eq:dyn-special} with 
$$I(t) = \left\{\begin{array}{ll} \mxt & \text{ for } t\in \{0,\ldots,k-2\}\\[4pt]
\Mxt & \text{ for } t\in \{k-1,\ldots,2k-3\} \end{array}\right.$$  	 
If $n<2k$ then
$$ \max_i x_i(T) - \min_i x_i(T) \leq \left( 1\!-\!\tfrac1k \right) \!\big( \max_i x_i(0) - \min_i x_i(0) \big)$$
where $T=2k\!-\!2$.
\end{lemma}
\begin{proof} 
For the sake of compactness, we set
$$\alpha(t) := \min_i x_i(t)\,,\quad \beta(t) := \max_i x_i(t)\,,\quad \gamma := \left(1-\tfrac1k\right),$$
introduce the two sequences
$$y(t) := \max_{i\in N_{\mu(\xx(t))}} x_i(t) \quad \text{and} \quad z(t) := \min_{i\in N_{M(\xx(t))}} x_i(t)\,,$$
and set $R=k\!-\!1$. We have
\begin{align*}
\beta\dkmo&-\alpha\dkmo = \beta\dkmo-z\dkmo+z\dkmo-\alpha\dkmo\\[3pt]
&\le \gamma\big(\beta\kmo\!-\!z\kmo\big)+z\kmo\!-\!\alpha\kmo
\intertext{using Lemma~\ref{lem:shrink-prelim-M} with initial configuration $\xx\kmo$. Then}
& = \gamma\big(\beta\kmo\!-\!y\kmo\big)+\gamma\big(y\kmo\!-\!z\kmo\big) +z\kmo\!-\!\alpha\kmo \\[3pt]
&\le \gamma\big(\beta\kmo\!-\!y\kmo\big)+\big(y\kmo\!-\!z\kmo\big) +z\kmo\!-\!\alpha\kmo
\intertext{since $\gamma<1$ and since  $y\kmo-z\kmo \ge0$ if $n<2k$ by Lemma~\ref{lem:z-y}. Then}
&= \gamma\big(\beta\kmo\!-\!y\kmo\big)+ y\kmo \!-\!\alpha\kmo \\[3pt]
&\le \gamma\big(\beta(0)-y(0)\big)+ \gamma\big(y(0) -\alpha(0)\big) \\[3pt]
& =  \gamma\big(\beta(0)-\alpha(0)\big) 
\end{align*}
{using Lemma \ref{lem:prelim-mu} and \ref{lem:shrink-mu} with initial configuration $\xx(0)$. We have finally obtained}
$\beta\dkmo-\alpha\dkmo \leq \gamma\big(\beta(0)-\alpha(0)\big)$.
\end{proof}


If $n<2k$ and the agent $I(t)$ that updates his opinion at time $t$ is chosen independently and uniformly at random over $V$, then process~\eqref{eq:dyn-special} converges almost surely to a consensus, from any initial configuration.
The almost sure convergence is guaranteed because the finite sequence of updates introduced in the Lemma~\ref{lem:special-sequence} appears infinitely often with probability one. This fact is proved in the following theorem, which provides the desired converge result.
\begin{theorem}\label{theo:asinc-conv-cons}
Let $n,k$ with $1\leq k\leq n$ be given. 
Let $\{I(t), t\ge0\}$ be a sequence of independent and uniformly distributed random variables over $\{1,\ldots,n\}$ and consider dynamics~\eqref{eq:dyn-special}.
If $n<2k$, then 
$$\lim_{t\to\infty} \xx(t) = \1 c \quad \text{ almost surely}$$ 
for any $\xx^0\in \R^n$, with $c \in [\min_i(x_i^0) , \max_i(x_i^0)]$.
\end{theorem}
\begin{proof} Let $\delta(t) = \max_i x_i(t) - \min_i x_i(t)$ and observe that, for any $\xx(0)=\xx^0$ and $\{I(t), t\ge0\}$,
$$\delta(0)\ge0  \quad \text{and} \quad 0\leq \delta(t+1) \leq\delta(t) \text{ for every }t\ge 0\,,$$  because the updates in the dynamics \eqref{eq:dyn-special}, based on model \eqref{eq:model}, involve convex combinations: the element with highest opinion cannot increase it and the element with lowest opinion cannot decrease it.
We introduce the sequence of events $\{A_t, t\ge 2k\!-\!3\}$ with 
\begin{align*} 
A_t = \big\{ &I(s) \!=\! \mu(\xx(s)) \text{ for } s\in \{t\!-\!2k\!+\!3,\ldots, t\!-\!k\!+\!1\} \text{ and } \\
&I(s) \!=\! M(\xx(s)) \text{ for } s\in \{t\!-\!k\!+\!2,\ldots, t\}  \big\}\,,
\end{align*}
i.e. the event $A_t$ is the occurrence of the finite sequence introduced in Lemma~\ref{lem:special-sequence} in the time window $\{t\!-\!(2k\!-\!3),\ldots,t \}$. 
In the same lemma we proved that, 
given the occurrence of $A_t$, 
we have $\delta(t\!+\!1) \leq (1-\frac1k)\,\delta(t\!-\!2k\!+\!3)$. Observe that 
$$0 \leq \lim_{t\to\infty} \delta(t) \leq \lim_{t\to\infty} \left(1-\tfrac1k\right)^{n_t}\delta(0)   $$
where $n_t$ is the number of times $A_t$ occurred up to time $t$. If $\P(A_t \text{ infinitely often})=1$ then $n_t \to \infty $ for $t\to \infty $ and the rightmost limit above is zero almost surely. Hence, $\lim_{t\to\infty} \delta(t)$ almost surely, which implies the convergence to consensus. Moreover, $c \in [\min_i(x_i^0) , \max_i(x_i^0)]$ because every update in \eqref{eq:dyn-special} is a convex combination of a subset of the current opinions. 

It remains to prove $\P(A_t \text{ infinitely often})=1$. The events of the sequence $\{A_t, t\ge 2k-3\}$ are not independent but the events in  the subsequence $\{A_{t_h}, h\ge 1\}$ where $t_h= h(2k-2) -1$ are. Each of these events has probability 
$$\P(A_{t_h}) = \left( \frac1n \right)^{2k-2} \,,$$
thus $\sum_{h=1}^\infty \P(A_{t_h}) = \infty$. 
Hence, 
$\{A_t \text{ i.o.}\} \supset \{A_{t_h} \text{i.o.}\}$.
From the second Borel-Cantelli lemma \cite[Ch.~2, Thm~18.2]{gut2012probability} 
$\P(A_t \text{ infinitely often}) \ge \P(A_{t_h} \text{ infinitely often})=1\,.$
\end{proof}
%
%
%
%
%
%
The result continues to hold for dynamics where $I(t)$ is not uniformly distributed over $\{1,\ldots,n\}$, as long as the probability to sample each agent is constant and positive. 
%
The proof has been based on exhibiting one suitable ``shrinking sequence'': however, it is clear that plenty of other sequences could do the job and actually play a role in inducing convergence of the dynamics. Therefore, the proof does not imply any good estimate of the convergence time. 

%


\section{Conclusion}


In this paper we have introduced a new model of opinion dynamics with opinion-dependent connectivity following the $k$-nearest-neighbors graph. 
The model is motivated by the rise of online social network services, where recommender systems select a certain number of news items to present to users, reducing the number of possible interactions to those which are closer to the user's presumed tastes.
The resulting dynamics is substantially different from comparable models in the literature, such as bounded-confidence models. One key difference is the inherent lack of reciprocity of the interactions, which makes all convergence analysis challenging. Another key difference is the robustness of the formed clusters, whose opinions are hard to sway by external leader nodes. This feature makes control approaches based on leadership, like~\cite{FD-SM-MJ:17}, unsuitable to $k$-nearest-neighbors dynamics.




%
%
%
%

%
%


\end{document}